\documentclass[journal]{IEEEtran}

\usepackage[colorlinks,
linkcolor=red,
anchorcolor=green,
citecolor=blue
]{hyperref} 
\usepackage{amsfonts,amssymb,amsmath,array}
\usepackage{latexsym}
\usepackage{CJK}
\usepackage{cite}
\usepackage{bm}
\usepackage{color, soul}
\usepackage{colortbl}

\usepackage{graphicx}
\usepackage{epstopdf}
\usepackage{epsfig}
\usepackage{subfigure} 
\usepackage{framed}
\usepackage{verbatim}  
\usepackage{color,soul}
\definecolor{aliceblue}{rgb}{0.94, 0.97, 1.0}
\definecolor{blizzardblue}{rgb}{0.67, 0.9, 0.93}
\definecolor{antiquebrass}{rgb}{0.8, 0.58, 0.46}
\definecolor{beaublue}{rgb}{0.74, 0.83, 0.9}
\sethlcolor{blizzardblue}
\usepackage{mathrsfs}
\usepackage{algorithm} 
\usepackage{algorithmic} 
\usepackage{booktabs}
\usepackage{textcomp}
\usepackage{multirow}
\usepackage{lettrine}
\usepackage[scaled=0.92]{helvet}

\usepackage{algorithm}
\usepackage{algorithmic}

\usepackage{tcolorbox}

\usepackage{amsthm} 
\newtheorem{lemma}{Lemma}
\newtheorem{remark}{Remark}

\begin{document}

\title{Towards Big data processing in IoT: Path Planning and Resource Management of UAV Base Stations in Mobile-Edge Computing System}

\author{Shuo~Wan, Jiaxun~Lu, Pingyi~Fan,~\IEEEmembership{Senior Member,~IEEE} and Khaled~B.~Letaief{*},~\IEEEmembership{Fellow,~IEEE}\\

\small
Tsinghua National Laboratory for Information Science and Technology(TNList),\\

Department of Electronic Engineering, Tsinghua University, Beijing, P.R. China\\
E-mail: wan-s17@mails.tsinghua.edu.cn, lujx14@mails.tsinghua.edu.cn, ~fpy@tsinghua.edu.cn\\
{*}Department of Electronic Engineering, Hong Kong University of Science and Technology, Hong Kong\\
Email: eekhaled@ece.ust.hk}

\maketitle

\begin{abstract}
Heavy data load and wide cover range have always been crucial problems for online data processing in internet of things (IoT). Recently, mobile-edge computing (MEC) and unmanned aerial vehicle base stations (UAV-BSs) have emerged as promising techniques in IoT.
In this paper, we propose a three-layer online data processing network based on MEC technique. On the bottom layer, raw data are generated by widely distributed sensors, which reflects local information. Upon them, unmanned aerial vehicle base stations (UAV-BSs) are deployed as moving MEC servers, which collect data and conduct initial steps of data processing. On top of them, a center cloud receives processed results and conducts further evaluation. As this is an online data processing system, the edge nodes should stabilize delay to ensure data freshness. Furthermore, limited onboard energy poses constraints to edge processing capability. To smartly manage network resources for saving energy and stabilizing delay, we develop an online determination policy based on Lyapunov Optimization. In cases of low data rate, it tends to reduce edge processor frequency for saving energy. In the presence of high data rate, it will smartly allocate bandwidth for edge data offloading. Meanwhile, hovering UAV-BSs bring a large and flexible service coverage, which results in the problem of effective path planning. In this paper, we apply deep reinforcement learning and develop an online path planning algorithm. Taking observations of around environment as input, a CNN network is trained to predict the reward of each action. By simulations, we validate its effectiveness in enhancing service coverage. The result will contribute to big data processing in future IoT.

\end{abstract}


%
\IEEEpeerreviewmaketitle

\section{Introduction}
%
%
%
%
\IEEEPARstart{T}{he} internet of things (IoT) has emerged as a huge network, which extends connected agents beyond standard devices to any range of traditionally non-internet-enabled devices. For instance, a large range of everyday objects such as vehicles, home appliances and street lamps will enter the network for data exchange. This extension will result in an extraordinary increase of required cover range and data amount, which is far beyond the existing network capability. For online data processing with delay requirement, the conventional cloud computing will face huge challenges. In order to collect and process big data sets with wide distribution, mobile-edge computing (MEC) and unmanned aerial vehicle base stations (UAV-BSs) have recently emerged to add existing networks with intelligence and mobility.

Conventionally, cloud computing has been deployed to provide a huge pool of computing resources for connected devices \cite{zhang2015toward}. However, as the data transmission speed is limited by communication resources, cloud computing can not guarantee its latency \cite{he2018integrated}. In face of high data rate in IoT, the data transmission load will overwhelm the communication network, which poses great challenge to online data processing. Recently, mobile-edge computing (MEC) has emerged as a promising technique in IoT. By deploying cloud-like infrastructure in the vicinity of data sources, data can be partly processed at the edge \cite{7883946}. In this way, the data stream in network will be largely reduced.

In existing works, the problem with respect to computation offloading, network resource allocation and related network structure designs in MEC have been broadly studied in various models \cite{he2018integrated,bi2018computation,rimal2017cloudlet,mao2017stochastic,park2016joint}. In \cite{he2018integrated}, the authors employed deep reinforcement learning to allocate cache, computing and communication resources for MEC system in vehicle networks. In \cite{bi2018computation}, the authors optimized the offload decision and resource allocation to obtain a maximum computation rate for a wireless powered MEC system. Considering the combination of MEC and existing communication service, a novel two-layer TDMA-based unified resource management scheme was proposed to handle both conventional communication service and MEC data traffic at the same time \cite{rimal2017cloudlet}. In \cite{mao2017stochastic}, the authors jointly optimized the radio and computational resource for Multi-user MEC computing system. In addition to the edge, the cloud was also taken into consideration in \cite{park2016joint}.

MEC system design considering computation task offloading has been sufficiently investigated in previous works. However, for IoT-based big data processing, MEC server may also serve to process local data at the edge \cite{shi2016edge,shi2016promise,8336572}. In \cite{shi2016edge}, the authors discussed the application of MEC in data processing. In \cite{shi2016promise}, the authors indicated that edge servers can process part of the data rather than completely send them to the cloud. Then in \cite{8336572}, the authors proposed a scheme for this system. In the field of edge computing, the research of edge data processing algorithm is still an open problem.

In IoT network, devices are often widely distributed with flexible movement. In this situation, conventional ground base station faces great challenge to provide sufficient service coverage. To figure out the problem, unmanned arial vehicle base stations (UAV-BSs) has recently emerged as a promising technique to add the network coverage with flexibility. In UAV-BSs wireless system, energy-aware UAV deployment and operation mechanisms are crucial for intelligent energy usage and replenishment \cite{DBLP:journals/corr/ZengZL16a}. In the literature, this issue has been widely studied \cite{matolak2015unmanned,mozaffari2016efficient,bor2016efficient,lu2018beyond,jeong2018mobile}. In \cite{matolak2015unmanned}, the authors characterized the UAV-ground channels. In \cite{mozaffari2016efficient}, the optimal hovering attitude and coverage radius were investigated. In \cite{bor2016efficient}, the authors jointly considered energy efficiency and user coverage to optimize UAV-BS placement. In \cite{lu2018beyond}, the authors considered the placement of UAV-BSs with the criterion of minimizing UAV-recall-frequency. Furthermore, UAV-BSs were also considered as a MEC server in \cite{jeong2018mobile}. However, they only considered one UAV-BS, focusing on the computation offloading problem. Besides, the cloud center was excluded from discussions.

\begin{figure}[tbp]
  \centering
  \includegraphics[width=0.9\columnwidth]{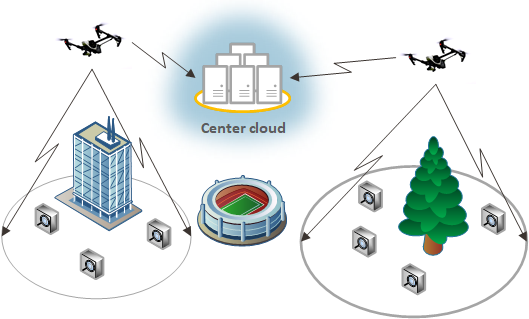}\\
  \caption{The structure of the three-layer network system. Data are generated by distributed sensors and transmitted to the cloud through UAVs at the edge. }\label{system_map}
\end{figure}

In IoT network, the data sets are generated by distributed sensors, which reflect their local information. In tasks such as supervision, the network is supposed to keep collecting and processing distributed data. Considering data freshness, the system should work in an online manner. In conventional cloud computing, all data will be transmitted to the cloud through base stations. Though the cloud may be powerful enough, the huge amount data will still pose a heavy load on the communication network. Furthermore, building base stations in a large region may cost too much, especially for rural regions. In this paper, we consider a MEC-based IoT network, where hovering UAV-BSs are deployed as edge servers. The network structure is shown in Fig. \ref{system_map}. The system is composed of three layers, involving distributed sensors, UAV-BSs and the center cloud. Distributed sensors keeps generating data, which is collected by nearby UAV-BSs. Each UAV-BS is equipped with onboard edge servers for executing initial steps of data processing. A large proportion of redundant data are split out and the extracted information is transmitted to the cloud for further analysis. The edge processing will largely relive the heavy burden on communication network. However, the limited edge computational capacity will bring new challenges. To balance the burden, part of the data will be directly offloaded to the cloud. The rest data will be temporarily stored in edge buffers, which results in delay. In this paper, it is assumed that the cloud is power enough. Therefore, our fucus is on the mobile edge nodes-UAV-BS, and discuss how to minimize the cost and delay at the edge.


The system design faces great challenges with respect to the cooperation of different layers and agents. In this paper, we investigate the problems related to UAV path planning and network resource management. Our major contributions are summarized as follows:

\begin{itemize}
 \setlength{\itemsep}{1pt}
 \setlength{\parskip}{0pt}
 \setlength{\parsep}{0pt}
 \item
 We propose a three-layer data processing network structure, which integrates cloud computing, mobile edge computing (MEC) and UAV base stations (UAV-BSs), as well as distributed IoT sensors. Data generated by distributed sensors are transmitted to UAV-BSs with onboard edge servers. It is assumed that redundant data are split out at the edge and the extracted information takes only a few bandwidth to transmit. In face of high data rate, the rest bandwidth will be allocated to UAV-BSs for data offloading. This system will largely relive the communication burden while providing a flexible service coverage.
 \item
 A reinforcement learning based algorithm is proposed for UAV-BS path planning. A local map of the around service requirement is taken as input to train a CNN neutral network, which predicts a reward for each possible action. The training samples are obtained by trials, feedbacks and corresponding observations. Considering heavy computational burden of network training, the training process is accomplished by powerful center cloud. Each UAV-BS receives network weights from cloud and selects its own moving action based on current local observations. By well-trained neutral network, UAV-BSs will automatically cooperate to cover the region of interest.
 \item
 The distributed online data processing system faces challenges in network management. As the onboard energy and computational resources of UAV-BSs are limited. In face of high data rate, part of received data will be offloaded to the cloud. Meanwhile, in face of low data rate, edge servers can lower down processor frequency for saving energy. Besides, they can also offload part of the data to further reduce energy consumption. This leads to the issue with respect to optimal network resource management. In this paper, we propose an online network scheduling algorithm based on Lyapunov optimization framework \cite{neely2010stochastic}. Without probability distributions of data sources, the network updates its policy by current buffer length, aimed at stabilizing delay while saving energy.
 \item
 The proposed algorithms are tested by simulations on Python. Simulation results show that the region of interest can be covered with good balance and high efficiency under our proposed path planning. Meanwhile, the performance with respect to energy consumption and delay are also tested in simulations. The results may assist to build an IoT network for processing a huge amount of data distributed in a large area.
\end{itemize}

The rest paper is organized as follows. We will introduce the system model and some key notations in Section \uppercase\expandafter{\romannumeral2}. In Section \uppercase\expandafter{\romannumeral3}, the path planning problem based on deep reinforcement learning will be investigated. In Section \uppercase\expandafter{\romannumeral4}, the network scheduling algorithm for data processing will be proposed based on Lyapunov optimization. The simulation results of data processing network will be shown in Section \uppercase\expandafter{\romannumeral5}. Finally, we will conclude in Section \uppercase\expandafter{\romannumeral6}.

%
\section{System model}

Consider an online distributed data processing network as shown in Fig .\ref{system_map}, where the data sources are $L$ distributed sensors denoted as $\mathbf{D}=\{d_{l}\}$. Upon them, $K$ hovering UAV-BSs carrying onboard edge servers are denoted as $\mathbf{U}=\{ u_{k} \}$. They collect data from around sensors and execute initial steps of data processing. The edge processing will split out a large sum of redundant data and the extracted information will be transmitted towards center cloud $C$ for further analysis. The internal environmental state is $\mathbf{S}=\{ s(t) \}$, which is affected by environmental elements and network scheduling policy. The observations of $s(t)$ by $u_{k}$ compose the set $\mathbf{O}(t)=\{ o_{k}(t) \}$. We denote the sensor index set as $\mathbb{L}=\{ 1, 2, ......, L \}$. The UAV-BS index set is $\mathbb{K}=\{ 1, 2, ......, K \}$. The system time set is $\mathbb{T}=\{ 0, 1, 2, ...... \}$, with interval $\tau$. In this section, we will introduce the network model, involving Air-Ground channel model, data generation model, UAV path planning model and edge computing model.

\subsection{Air-ground channel}
The Air-Ground (AG) channel involves line-of-sight (LOS) link and non-line-of-sight (NLOS) link \cite{matolak2015unmanned}. In the literature \cite{al2014optimal}, the corresponding pass loss is defined as follows.

\begin{equation}\label{pl1}
L_{\xi}(r,h)=\left\{
\begin{array}{rcl}
(\frac{4\pi f}{c})^{2}(r^{2}+h^{2})\eta_{0} & & {\xi=0}\\
(\frac{4\pi f}{c})^{2}(r^{2}+h^{2})\eta_{1} & & {\xi=1}\\
\end{array} \right.
\end{equation}
where $\xi=0$ and $\xi=1$ separately represents LOS link and NLOS link. Projecting the UAV on the ground, its distance from the covered sensor is denoted as $r$. Besides, $c$ is the speed of light and $f$ represents the signal frequency. Parameter $h$ is the hovering altitude of UAV-BSs, while $\eta_{0}$ and $\eta_{1}$ are respectively the path loss parameters for LOS link and NLOS link. As obstacles will typically reduce a large proportion of signal intensity, we have $\eta_{0}<<\eta_{1}$.

The probability of LOS link is affected by environmental elements, which is given by \cite{al2014optimal} as
\begin{equation}
p_{0}(r,h)=\frac{1}{1+aexp(-b[\theta-a])}
\end{equation}
where $a$ and $b$ are environmental constants of the target region and $\theta=\tan^{-1}(\frac{h}{r})$ is the elevation angle of UAV-BSs. Meanwhile, $1-p_{0}(r,h)$ represents the NLOS probability. Then the final average path loss of AG channel is
\begin{equation}\label{pl3}
\overline{L}(r,h) =p_{0}(r,h)L_{0}(r,h)+(1-p_{0}(r,h))L_{1}(r,h)
\end{equation}
\begin{table}
  \centering
  \caption{Summary of key notations}\label{tb1}
  \begin{tabular}{c m{6cm} c}
  \hline
    Notation& Explanations\\
    \hline
    $\mathbf{D}=\{d_{l}\}$& Set of distributed sensors\\
    \hline
    $\mathbf{U}=\{ u_{k} \}$& Set of UAV-BSs\\
    \hline
    $\mathbf{S}=\{ s(t) \}$ & The internal environmental state in time slot $t$\\
    \hline
    $\mathbf{O}(t)=\{ o_{k}(t) \}$& Set of observations of local environmental elements by UAV-BSs\\
    \hline
    $\mathbb{T}$ & Set of system time slot\\
    \hline
    $\mathbb{T}_{p}$ & Set of time slot for UAV path update\\
    \hline
    $\mathbf{p}_{k}(t_{p})$ & The position in planned path for UAV $u_{k}$ at time slot $t_{p}$\\
    \hline
    $\mathbf{v}_{k}(t_{p})$& The path update policy of UAV $u_{k}$ at time slot $t_{p}$\\
    \hline
    $A_{l}(t)$& The generated data bits of sensor $d_{l}$ in time slot $t$\\
    \hline
    $A_{u,k}(t)$& The collected data bits by UAV $u_{k}$ in time slot $t$\\
    \hline
    $D_{l,k}(t)$ &The capability of edge data processing on $u_{k}$ in time slot $t$\\
    \hline
    $D_{tm,k}(t)$ &The capability of data transmission through network in time slot $t$\\
    \hline
    $Q_{k}(t)$ & The edge buffer length on $u_{k}$ at $t$\\
    \hline
    $f_{k}(t)$ & The edge processor frequency on $u_{k}$ at $t$\\
    \hline
    $p_{tm,k}(t)$ &The data transmission power of $u_{k}$ at $t$\\
    \hline
    $a_{k}(t)$& The proportion of allocated bandwidth to $u_{k}$ at $t$\\
    \hline
    $\alpha$ &The update rate of network training\\
    \hline
    $\rho_{j}$& The occurring frequency of action $j$ in training samples\\
    \hline
    $\gamma$ & Decay coefficient of future rewards\\
    \hline
    $\phi_{l}$ &A coefficient reflecting the uncover rate of sensor $d_{l}$\\
\hline
\end{tabular}

\end{table}

\subsection{UAV path}
The position of $u_{k}$ is denoted as $[x_{k},y_{k},h_{k}]$, where $\mathbf{p}_{k}=[x_{k},y_{k}]$ represents its projection on the ground and $h_{k}$ is its corresponding hovering altitude. It is assumed that $u_{k}$ covers sensors around $\mathbf{p}_{k}$ within radius $r$. In our previous work \cite{lu2018beyond}, we proved that the optimal height $h_{k}^{*}$ satisfies
\begin{equation}
h_{k}^{*}={\rm tan}(\theta_{b}^{*})r
\end{equation}
where $\theta_{b}^{*}$ is the optimal elevation angle on the coverage boundary. That is, ${\rm tan}(\theta_{b}^{*})$ is the optimal height with $r=1$. It is assumed that the data transmission rate is $C$ and the channel path loss is modeled as the above sub-section. In this case, ${\rm tan}(\theta_{b}^{*})$ can be derived by binary research, see \cite{lu2018beyond}.
By optimized $h_{k}^{*}$, the UAV path only involves two-dimensional position $\mathbf{p}_{k}=[x_{k},y_{k}]$. The time slot for path update is $t_{p} \in \mathbb{T}_{p}$ with interval $\tau_{p}$, where $\mathbb{T}_{p}$ is the time slot set for path update. The corresponding position is denoted as $\mathbf{p}_{k}(t_{p})=[x_{k}(t_{p}),y_{k}(t_{p})]$. Note that the reaction speed of flight control system is typically slower than computation and communication management. While $\tau$ is typically tiny, $\tau_{p}$ should be larger than $\tau$.

In this paper, the UAV path update is conducted in an online manner. At $t_{p}$, the path node for next time slot is determined based on observation set $\mathbf{O}(t_{p})=\{ o_{k}(t_{p}) \}$. Suppose the position of $u_{k}$ at $t_{p}$ is $\mathbf{p}_{k}(t_{p})=[x_{k}(t_{p}),y_{k}(t_{p})]$, its position in path at $(t_{p}+1)$ is
\begin{equation}\label{u_path}
\mathbf{p}_{k}(t_{p}+1)=\mathbf{p}_{k}(t_{p})+[v_{k,x}(t_{p}),v_{k,y}(t_{p})]
\end{equation}
where $\mathbf{v}_{k}(t_{p})=[v_{k,x}(t_{p}),v_{k,y}(t_{p})] \in \mathbb{V}$ is the path update part for time slot $t_{p}$. $\mathbb{V}$ is the candidate policy set. Therefore, the path $\mathbf{P}_{k}$ for $u_{k}$ is
\begin{equation}\label{d_path}
\{ \mathbf{p}_{k}(t_{p}) | t_{p} \in \mathbb{T}_{p}, \mathbf{p}_{k}(t_{p}+1)=\mathbf{p}_{k}(t_{p})+\mathbf{v}_{k}(t_{p}), \mathbf{v}_{k}(t_{p})\in \mathbb{V} \}
\end{equation}
The entire multi-UAV path set is denoted as $\mathbf{P}=\{ \mathbf{P}_{k} | k\in \mathbb{K} \}$.

\subsection{Data generation}
The distributed sensors generate data involving local information. The data is temporarily stored in its buffer denoted as $b_{l}$. It is assumed that sensor $d_{l}$ generates $A_{l}(t)$ bits data during time slot $t$, where $t \in \mathbb{T}$. Parameter $A_{l}(t)$ is an i.i.d. random variable. It is supposed that $A_{l}(t)$ satisfies poisson distribution with $E(A_{l}(t)) = \lambda_{l}$. In practical systems, $A_{l}(t)$ is typically constrained by hardware limitation. Therefore, $A_{l}(t)$ is assumed to be bounded by $[0, A_{max}]$, where $A_{max}$ is the largest value of $A_{l}(t)$. Note that $\lambda_{l}$ is an empirical parameter which may vary among different places.

\subsection{Edge computing}
It is assumed that data collection and its correlated network scheduling policy are updated in discrete time slots with interval $\tau$ \cite{zeng2016throughput,jeong2018mobile}. We suppose $u_{k}$ collects $A_{u,k}(t)$ bits data in time slot $t$. The collected data will be temporarily stored in edge data buffer.

Initial steps of data processing are executed at the edge, where a large amount of redundant data is split out. It is supposed that the extracted information at the edge takes only part of the bandwidth between edge and cloud for transmission. This relieves the heavy burden on network communication. However, the limited edge processing capability will bring new challenges.  In this case, the rest bandwidth can be allocated to edge nodes for data offload, which balances the burden on edge processing and network communication.
\subsubsection{Data caching}
In time slot $t$, the data processing capability on $u_{k}$ is $D_{l,k}(t)$, while the edge data offloading capability is $D_{tm,k}(t)$. The queuing length at the beginning of time slot $t$ on $u_{k}$ is $Q_{k}(t)$, which evolves as follows.
\begin{align}\label{qup}
Q_{k}(t+1)={\rm max}\{ Q_{k}(t)+A_{u,k}(t)\nonumber\\-D_{l,k}(t)-D_{tm,k}(t), 0 \}
\end{align}
where $Q_{k}(0)$ is set to be zero.

\subsubsection{Edge processing}
It is assumed that the edge server on $u_{k}$ needs $L_{k}$ CPU cycles to precess one bit data, which depends on the applied algorithm \cite{mao2017stochastic}. The CPU cycle frequency of $u_{k}$ in time slot $t$ is denoted as $f_{k}(t)$, where $f_{k}(t) \in [0,f_{max}]$. Then $D_{l,k}(t)$ is
\begin{equation}
D_{l,k}(t)=\frac{\tau f_{k}(t)}{L_{k}}
\end{equation}
where $\tau$ is the time slot length. The power consumption of edge data processing \cite{burd1996processor} by $u_{k}$ is
\begin{equation}
p_{l,k}(t)=\kappa _{k}f_{k}^{3}(t)
\end{equation}
where $\kappa _{k}$ is the effective switched capacitance \cite{burd1996processor} of $u_{k}$, which is determined by processor chip structure.

\subsubsection{Data offloading}
It is assumed that the wireless channels between UAV-BSs and center cloud are i.i.d. frequency-flat block fading \cite{lu2018beyond}. Thus the channel power gain between $u_{k}$ and center cloud is supposed to be $\Gamma_{k}(t)=\gamma_{k}(t)g_{0}(\frac{d_{0}}{d_{k}})^{\theta}$, where $\gamma_{k}(t)$ represents the small-scale fading part of channel power gain, $g_{0}$ is the path loss constant, $\theta$ is the path loss exponent, $d_{0}$ is reference distance and $d_{k}$ is the distance between $u_{k}$ and center cloud. Let us consider the system working in FDMA mode, the data transmission capacity from $u_{k}$ to center cloud is
\begin{equation}
D_{tm,k}(t)=a_{k}(t)W\tau{\rm log}_{2}(1+\frac{\Gamma_{k}(t)p_{tm,k}(t)}{a_{k}(t)N_{0}W})
\end{equation}
where $a_{k}(t)$ is the proportion of the bandwidth allocated to $u_{k}$, $p_{tm,k}(t)$ is the transmission power with $p_{tm,k}(t)\in [0,P_{max}]$, $W$ is the entire bandwidth for data offloading and $N_{0}$ is the power spectral density of noise.

\section{UAV Path Planning}
Moving UAV-BSs provide a flexible and wide service coverage, which is especially effective for surveillance tasks. However, all the advantages must be built on smart path planning. In \cite{jeong2018mobile}, the authors proposed an off-line path planning algorithm based on convex optimization. However, it only aims at a single UAV. In multi-UAV system, there exists correlation among UAV-BSs. Each UAV-BS may only obtain local observations. Furthermore, many unexpected environmental factors may pose great challenge to off-line path planning. Therefore, it is essential to adaptively plan UAV path in an online manner.

In the last decade, deep reinforcement learning has obtained impressive results in online policy determination. Different from conventional reinforcement learning, deep reinforcement learning trains deep neutral network to predict rewards of each candidate action. Typically, the neutral network is utilized to fit complex unknown functions in learning tasks \cite{lecun2015deep}. Besides, it can handle more complex input features. In \cite{mnih2015human}, the authors adopted deep reinforcement learning to train a CNN network for playing computer games with online policy. In this paper, we adopt a similar way to train an adaptive path planning network. For $u_{k}$ at time $t$, its input is observation $o_{k}(t)$. In this section, we will discuss the problem formulation and its solution based on deep reinforcement learning.

\subsection{Problem formulation}

The UAV path is planned in terms of time slot $t_{p}$. Our objective is to optimize $\mathbf{P}$ to enhance UAV coverage. In time slot $t_{p}$, $u_{k}$ is supposed to use the plan $\mathbf{p}_{k}(t_{p}+1)$ by local observation $o_{k}(t_{p})$. The policy is determined in a distributed manner without global information. However, local $o_{k}(t_{p})$ is not sufficient to depict the entire coverage. In this case, we need to find an alternative optimization objective to represent entire UAV coverage. Typically, an ideal coverage will sufficiently utilize data processing capability of $u_{k}$. That is, if UAV-BSs cooperate to enhance data collection amount, they will achieve a relatively good coverage. Therefore, the path planning problem is formulated as follows.

We suppose $u_{k}$ collects $A_{u,k}(t_{p})$ bits data in time slot $t_{p}$. It is straightforward to see $A_{u,k}(t_{p})$ is determined by state set $\mathbf{S}$ and UAV path set $\mathbf{P}$ within time slot $t_{p}$. The connection is represented by
\begin{equation}
A_{u,k}(t_{p})=f_{t_{p}}(\mathbf{S},\mathbf{P})
\end{equation}
where $f_{t_{p}}$ is a time varying function determined by environmental elements. The environmental state is supposed to be characterized by a Markov process. The state update is determined by current state $s(t_{p})$ and path set $\mathbf{P}$, which is represented by
\begin{equation}
s(t_{p}+1)=g(s(t_{p}),\mathbf{P})
\end{equation}
Then the problem is formulated as follows.
\begin{flalign}\label{equ:data1}
\mathcal{P}_{1\text{-}\rm{A}}:\,\,\max_{\mathbf{P }} \,\,\, & \frac{1}{\left | \mathbb{T}_{p} \right |} \sum_{t_{p}\in \mathbb{T}_{p}}\frac{1}{K}\sum_{k=1}^{K}A_{u,k}(t_{p})\\
\textrm{s.t.}\,\,\,& \mathbf{p}_{k}(t_{p}+1)=\mathbf{p}_{k}(t_{p})+\mathbf{v}_{k}(t_{p}), \mathbf{v}_{k}(t_{p})\in \mathbb{V}.\tag{\theequation a}\label{p1a_a}\\
&s(t_{p}+1)=g(s(t_{p}),\mathbf{P}).\tag{\theequation b}\label{p1a_b}\\
&A_{u,k}(t_{p})=f_{t_{p}}(\mathbf{S},\mathbf{P}).\tag{\theequation c}\label{p1a_c}
\end{flalign}
where constraint (\ref{p1a_a}) represents the path update policy. Constraint (\ref{p1a_b}) represents the internal state update, which is determined by specific environment. Constraint (\ref{p1a_c}) represents the system reward by $\mathbf{S}$ and $\mathbf{P}$.

The direct optimization of $\mathcal{P}_{1\text{-}\rm{A}}$ faces great challenges. In multi-agent system, there exists correlation among agents. Models in (\ref{p1a_b}) and (\ref{p1a_c}) are determined by complex environmental elements involving correlations among UAV-BSs. Therefore, it is very hard to specifically model $g$ and $f_{t_{p}}$. Furthermore, the internal environmental state $\mathbf{S}$ is also beyond our reach. Instead, we can only plan path by local observation $o_{k}(t_{p})$. In this case, training an alternative function to approximate the complex environmental models may provide an achievable solution. This is the so-called reinforcement learning algorithm.

\subsection{Reinforcement learning algorithm}
The optimal policy is selected by rewards of each candidate action. In reinforcement learning, the Q-function $Q(s,a)$ represents the rewards $r(t)$ of action $a$ under state $s$. Faced with complex environmental elements, it is very hard to model Q-function specifically. In this case, reinforcement learning is applied to learn $Q(s,a)$ by iteratively interacting with around environment. By trials and feedbacks, they will obtain training samples in form of $(s(t),a(t),s(t+1),r(t))$. With these dynamically updating training samples, the trained $Q(s,a)$ will be a good approximation to the environmental Q-function. Reinforcement learning enables agents to learn an adaptive policy maker, which is widely applied in dynamic control and optimization. In path planning problem, UAV-BSs only obtain observations $o_{k}(t_{p})$ of internal state $s(t_{p})$. To explore internal features in obtained observations, deep Q-learning algorithm is applied.


In deep-Q-learning, a deep neutral network $Q(o,a,\theta)$ is applied to approximate Q-function, where $\theta$ represents network weights and $o$ is the observation data. Taking $o$ as input, the Q-network will output predicting rewards of each candidate action. By continuous interaction with around environment, $Q(o,a,\theta)$ will be adaptively adjusted to fit the unknown environmental model. In \cite{mnih2015human}, a CNN network is trained to adaptively play computer games with screen pictures as input. For such rather complex tasks, the observations can be matrix or sequence. In this case, the CNN neutral network can exploit local correlations of elements in $o$ by convolutional filters, which enables extractions of high-dimensional features. In many practical applications, the algorithm works robustly with high-level performance. The training process is summarized in Algorithm \ref{alg11}.

\begin{algorithm}
\renewcommand{\algorithmicrequire}{\textbf{Initialization:}}
\caption{Deep Q-learning process for UAV path planning}
\label{alg11}
\begin{algorithmic}
\REQUIRE
Initialize the relay memory $E(0)$; Initialize deep Q-network weights $\theta$; Initialize the reference network weights $\theta^{-}$ by $\theta$; Initialize $\{ \rho_{j} \}$, $\alpha$, $\gamma$ and $\alpha_{max}$.

\FOR{each epoch}
\STATE Randomly initialize UAV positions.
\WHILE{$t_{p}\leq T_{p}$}
\FOR{each $u_{k}$}
\STATE Collect around service requirements and generate observations $o_{k}(t_{p})$.
\STATE Randomly generate $\epsilon (t_{p}) \in [0,1]$.
\STATE Choose action $a(t_{p})$ by:
\IF{$p<\epsilon(t_{p}) $}
\STATE randomly select an action $a(t_{p})$
\ELSE
\STATE $a(t_{p})={\rm arg max}_{a}Q(o_{k}(t_{p}),a,\theta)$
\ENDIF
\STATE Move along the planned path by executing $a(t_{p})$.
\STATE Collect data from covered sensors.
\STATE Obtain the reward $r(t_{p})$ and observations $o_{k}(t_{p}+1)$.
\STATE Transmit $e(t_{p})=(o_{k}(t_{p}),a(t_{p}),r(t_{p}),o_{k}(t_{p}+1))$ to the central relay memory.
\ENDFOR
\STATE Randomly choose a batch of interaction experience $(o_{i},a_{i},r_{i},o_{i+1})$ from relay memory $E(t_{p})$.
\STATE Determine update rate $\alpha$ by (\ref{upalpha}) and calculate the target value $y_{i}$ by
$$y_{i}=\alpha (r_{i}+\gamma{\rm max}_{a^{'}}Q(o_{i+1},a^{'},\theta^{-}))+(1-\alpha)Q(o_{i},a_{i},\theta_{i})$$
\STATE Train the CNN neutral network $Q(o,a,\theta)$ by loss function $L(\theta)$ (\ref{lf}).
\STATE Update the reference network weights $\theta^{-}$ by $\theta$ every $G$ steps.
\STATE Update $\{ \rho_{j} \}$.

\ENDWHILE
\ENDFOR

\end{algorithmic}
\end{algorithm}

In the training process, the training samples generated by $u_{k}$ at $t_{p}$ is denoted as $e(t_{p})=(o_{k}(t_{p}),a(t_{p}),r(t_{p}),o_{k}(t_{p}+1))$, where $o(t_{p})$ represents the observations by $u_{k}$ at $t_{p}$, $a(t_{p})$ is its action, $r(t_{p})$ is the feedback reward and $o(t_{p}+1)$ is the new observations. In this paper, a central training mode is applied. Training samples of distributed UAV-BSs are gathered by center for network training. The UAV-BSs share the centrally trained network weights. Based on different local observations, they can choose separated actions. The collected training samples are stored in relay memory $E(t_{p})=\{ e(t_{p}-E+1), ......, e(t_{p}) \}$, where $E$ is the buffer length. Each time, the algorithm will randomly sample a batch from $E(t_{p})$ for training. Compared with conventional training by consecutive samples, this method may enable networks to learn from more various past experiences rather than concurrent experiences.

The MSE-based loss function $L(\theta)$ for $(o,a,r,o^{'})$ is defined as follows.
\begin{equation}\label{lf}
L(\theta_{i})=E[ (y_{i}-Q(o,a,\theta_{i}))^{2} ]
\end{equation}
where $y_{i}=\alpha (r+\gamma{\rm max}_{a^{'}}Q(o^{'},a^{'},\theta^{-}))+(1-\alpha)Q(o,a,\theta_{i})$ and $\theta^{-}$ is the reference network weight. Parameter $\gamma$ is the decay coefficient of future rewards while $\alpha$ is the update rate. Note that the loss for other actions in the policy set is set to be $0$.

To ensure convergence, $\alpha$ is typically set as $\frac{1}{\sqrt{t_{p}}}$. Note that the rather frequent action will be trained more tensely, which will break the balance among all candidate actions. Therefore, the sample proportion of each candidate action is maintained here, denoted as $\{ \rho_{j} \}$. Parameter $j$ is the action index. Suppose the sample action index is $j$ and $\alpha$ is upper-bounded by $\alpha_{max}$, $\alpha$ is determined by
\begin{equation}\label{upalpha}
\alpha={\rm min}\{ \alpha_{max},\frac{1}{\rho_{j}\sqrt{t_{p}}} \}
\end{equation}
where $\alpha_{max}$ is the maximum value of $\alpha$. Note that an action with a larger $\rho_{j}$ will have a smaller update rate.

\subsection{Interaction with environment}

\begin{figure}[tbp]
  \centering
  \includegraphics[width=0.99\columnwidth]{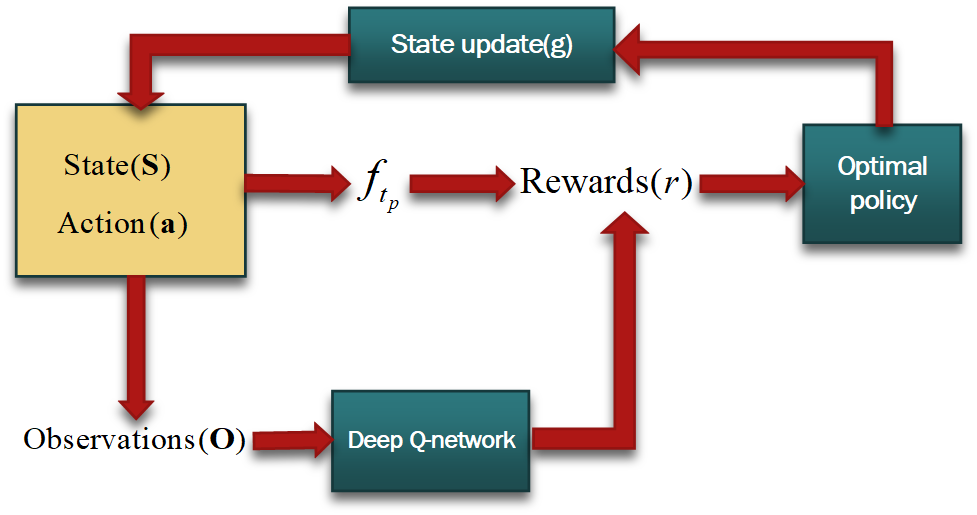}
  \caption{The interaction mode between deep Q-learning algorithm and environment.}\label{Q-implement}
\end{figure}
The environmental model $\{f_{t_{p}},g\}$ and the internal state $\mathbf{S}$ is unknown. In previous subsection, we proposed a deep Q-learning algorithm to adaptively learn environmental elements. Before its implementation, the specific interaction mode with around environment will be discussed in this subsection.

A model of the internal environment and its interaction with the deep Q-learning algorithm is shown in Fig. \ref{Q-implement}. Based on state $\mathbf{S}$ and action $a$, the internal environment will generate a reward $r$ by model $f_{t_{p}}$. In this case, an optimal policy is generated by maximizing the outcome rewards. Then the environmental state will be updated by its internal model $g$. To approximate this environmental model for policy learning, a deep Q-network is implemented to interact with the environment. The observations $\mathbf{O}$ is obtained by Q-network as input, which carries the essential information about $r$ within $\mathbf{S}$. By directly receiving the outcome $r$ from environment, the Q-network will be trained to adaptively estimate $r$. Based on its estimation, we will derive a nearly optimal policy. In this paper, a model-free reinforcement learning is applied. Therefore, the Q-network only needs to receive observations and estimate $r$, without considering the internal state update model $g$. The key elements of the interaction are observations, rewards and action policy.

%

\subsubsection{Observations}
The observations of distributed sensors should involve information of around service requirement, so that the planned path can ensure a better coverage. The sensors which have long been uncovered should have more urgent service requirement. Besides, sensors with larger data rate also requires more coverage. Furthermore, it is also important to avoid overlap among coverage of different UAV-BSs. Therefore, the observations by UAV-BSs should involve the above essential elements for a proper path.

It is straightforward to see that the local observations should be a two-dimensional data set.
Suppose at time $t_{p}$, the local observations involves a $M \times M$ region around $u_{k}$. The observation data is set as a $R \times R$ matrix $O=\{ o_{i,j} \}$. The position of $u_{k}$ is $\mathbf{p}_{k}=(x_{k},y_{k})$. Then the position in map corresponding to $o_{i,j}$ is $\mathbf{p}_{k,i,j}=(x_{k,i,j},y_{k,i,j})=(x_{k}-M+\frac{M}{R}(i-1),y_{k}-M+\frac{M}{R}(j-1))$. $o_{i,j}$ represents observations of sensors around $p_{k,i,j}$. In this way, the local region is represented in a discrete manner. Parameter $R$ is determined by the input data size of the Q-network $Q(o,a,\theta)$. $M$ is set according to the observation range of UAV-BSs. $\frac{M}{R}$ is called the observation sight, which describes the observation wideness.

We suppose sensor $d_{l} \in \mathbf{D}$ maintains its service requirement $\Phi_{l}$, which illustrates its data freshness and accumulation. The process is summarized in Algorithm \ref{alg12}. $\phi_{l}$ represents the data freshness of $d_{l}$. Local data rate $\lambda_{l}$ represents data accumulation rate. They are synthesized by $\Phi_{l}=\lambda_{l}\phi_{l}$. The initial sensor buffer $b_{l}(0)$ is supposed to be $0$. They are updated in terms of time slot $t\in \mathbb{T}$. If uncovered, the data freshness will decay by (\ref{alg2-1}). If covered by UAV-BSs, it is assumed that $d_{l}$ will transmit at most $B_{l}(t)$ bits data in time slot $t$. In this case, $b_{l}(t)$ will update by (\ref{alg2-2}) and the data freshness will be renewed by (\ref{alg2-3}).

It is assumed that $u_{k}$ can obtain $\Phi_{l}$ from the sensors in the $M \times M$ region around it. The processing of the corresponding observations is summarized in Algorithm \ref{alg13}. Matrix $O$ is initialized as zero matrix. $\Phi_{l}$ from sensors around $\mathbf{p}_{k,i,j}$ is added to $o_{i,j}$. In this way, $o_{i,j}$ will reflect the local data freshness and accumulation. For $\mathbf{p}_{k,i,j}$ covered by other UAV-BSs, $o_{i,j}$ will be adjusted by (\ref{alg3-1}). In this case, the observations will involve the coverage overlap among UAV-BSs. Note that $\mathbf{p}_{k,i,j}$ outside the region will lead to $o_{i,j}=0$. The processed $O=\{ o_{i,j} \}$ will be taken as input of the CNN Q-network for rewards estimation.

\begin{algorithm}
\renewcommand{\algorithmicrequire}{\textbf{Initialization:}}
\caption{Sensor data freshness maintaining process}
\label{alg12}
\begin{algorithmic}
\REQUIRE
Initialize $\phi_{l}$ and $b_{l}(0)$ as $0$.

\WHILE {$t \leq T$}
\STATE $b_{l}(t)=b_{l}(t)+A_{l}(t)$
\IF{$d_{l}$ is beyond coverage}
\STATE \begin{equation}\label{alg2-1}\phi_{l}=\phi_{l}+1 \end{equation}
\ELSE
\STATE \begin{equation}\label{alg2-2} b_{l}(t+1)={\rm max}\{ 0,b_{l}(t)-B_{l}(t) \} \end{equation}
\STATE \begin{equation}\label{alg2-3}\phi_{l}=\phi_{l} \frac{b_{l}(t+1)}{b_{l}(t)}\end{equation}
\ENDIF
\STATE $\Phi_{l}=\lambda_{l}\phi_{l}$
\ENDWHILE

\end{algorithmic}
\end{algorithm}

\begin{algorithm}
\renewcommand{\algorithmicrequire}{\textbf{Initialization:}}
\renewcommand{\algorithmicensure}{\textbf{Processing:}}
\caption{UAV-BSs observation processing on $u_{k}$ at $t_{p}$}
\label{alg13}
\begin{algorithmic}
\REQUIRE
Initialize $O=\{ o_{i,j} \}$ by zero matrix; Obtain position $p_{k}(t_{p})$; Observe $\Phi_{l}$ of distributed sensors in its around $M \times M$ region.

\ENSURE

\STATE Obtain the position $\mathbf{p}_{k,i,j}$ corresponding to $o_{i,j}$ as $$ \mathbf{p}_{k,i,j}=(x_{k}-M+\frac{M}{R}(i-1),y_{k}-M+\frac{M}{R}(j-1)) $$
\STATE Find $\mathbf{p}_{k,i,j}$ around sensor $d_{l}$.
\STATE Update $o_{i,j}$ corresponding to the above $\mathbf{p}_{k,i,j}$ by $$o_{i,j}=o_{i,j}+\Phi_{l}$$

\IF {$\mathbf{p}_{k,i,j}$ is covered by other $\widetilde{K}$ nearby UAV-BSs}
\STATE \begin{equation}\label{alg3-1}o_{i,j}=o_{i,j}-\Phi_{u}\widetilde{K}\end{equation}
\ENDIF

\end{algorithmic}
\end{algorithm}

%
%
%

\subsubsection{Action policy}
The path $\mathbf{P}_{k}$ for $u_{k}$ is defined by (\ref{d_path}). The corresponding action policy $a(t_{p})$ for online path planning is $\mathbf{v}_{k}(t_{p})=[v_{k,x}(t_{p}),v_{k,y}(t_{p})] \in \mathbb{V}$ defined in (\ref{u_path}). In this paper, we define a set $\mathbb{V}$ with finite candidate policy. It is assumed that $\left \| \mathbf{v}_{k}(t_{p}) \right \|$ is a constant. That is, the UAV speed is supposed to remain stable and the length of path update does not change.
Then the policy set with discrete direction is defined as follows.
\begin{equation}\label{equ:policy}
\Psi =\{ (v{\rm cos}(\theta_{b}),v{\rm sin}(\theta_{b}))|\theta_{b}=\frac{m\pi}{4}, m=0,1,...,7 \} \cup \{\overrightarrow{0} \}
\end{equation}
where $v$ is the length of a path step and $\theta_{b}$ is the discrete path angle. The zero element means hovering at the current position.

\subsubsection{Reward function}
The objective of $\mathcal{P}_{1\text{-}\rm{A}}$ is to maximize the overall data collection, so that the edge capability is sufficiently utilized. For distributed online decision, the reward must be accessible at the edge UAV-BSs. Therefore, the reward $r(t_{p})$ is defined as the collected data bits in time slot $t_{p}$. Note that the interaction experiences will be transmitted to center for network training. Furthermore, the observations also involve other around UAV-BSs. Therefore, in the process of interaction and learning, the UAV-BSs will tend to cooperate with each other to ensure a relatively good coverage.

\section{System Data Management}
After receiving data from around sensors, the UAV-BSs process their collected raw data and transmit the edge processing result to center cloud. It is assumed that the transmission of processing result takes very little communication resources. Therefore, the majority communication bandwidth between UAV-BSs and center cloud can be utilized for transmitting part of the unprocessed data. In this way, the edge system can enhance its data throughput while reducing UAV onboard energy cost. In this section, we will formulate the data offloading problem into a Lyapunov optimization problem. As the cloud is supposed to be powerful enough, we may consider the edge energy cost and data processing delay as system cost.

\subsection{Problem formulation}
The data offloading policy focus on stabilizing delay while reducing the power consumption of edge processing and data transmission. It is managed in terms of system time slot $t$. It is assumed that each UAV-BS is hovering at a constant speed. Thus, the power consumption of onboard dynamical system is excluded. At time slot $t$, the power consumption of local computation on UAV-BS $u_{k}$ is $p_{l,k}(t)$. The data transmission power of $u_{k}$ is $p_{tm,k}(t)$. We denote the power consumption of $u_{k}$ in time slot $t$ as
\begin{equation}
P_{k}(t)=p_{l,k}(t)+p_{tm,k}(t)
\end{equation}
Then the average weighted sum power consumption is
\begin{equation}\label{sump}
\overline{P}=\underset{T\rightarrow \infty }{lim}\frac{1}{T}\sum_{t=1}^{T}E\left [ \sum_{k=1}^{K}w_{k}P_{k}(t) \right ]
\end{equation}
where $w_{k}$ is a positive parameter with regard to $u_{k}$, which can be adjusted to balance power management of all UAV-BSs. As the system performance metrics, $\overline{P}$ is the long-term edge power consumption. The data offloading policy with respect to $\overline{P}$ can be derived by statistical optimization.

The data collected by $u_{k}$ will be temporarily stored in the onboard data buffer for future processing. In this case, the data queuing delay is the metrics of edge system service quality. By Little's Law \cite{little1961proof}, the average queuing delay of a queuing agent is proportional to the average queuing length. Therefore, the average data amount in onboard data memory is viewed as the system service quality metrics. The long-term queuing length for edge $u_{k}$ is defined as
\begin{equation}
\overline{Q}_{k}=\underset{T\rightarrow \infty }{lim}\frac{1}{T}\sum_{t=1}^{T}E[Q_{k}(t)]
\end{equation}

The network policy at time slot $t$ for $K$ UAV-BSs is denoted as $\mathbf{\Phi }(t)=[\mathbf{f}(t),\mathbf{p}_{tm}(t),\mathbf{a}(t)]$. The operation $\mathbf{f}(t)=\{ f_{1}(t), ......, f_{K}(t) \}$ is the processor frequency for edge data processing on UAV-BSs. The operation $\mathbf{p}_{tm}(t)=\{ p_{tm,1}(t), ......, p_{tm,K}(t) \}$ is the transmission power of data offloading. $\mathbf{a}(t)=\{ a_{1}(t), ......, a_{K}(t) \}$ is the proportion of bandwidth allocation among the $K$ UAV-BSs. Therefore, the optimization of edge data processing policy can be formulated as problem $\mathcal{P}_{2\text{-}\rm{A}}$.

\begin{flalign}\label{equ:data1}
\mathcal{P}_{2\text{-}\rm{A}}:\,\,\min_{\mathbf{\Phi }(t)} \,\,\, & \overline{P}\\
\textrm{s.t.}\,\,\,& \sum_{k=1}^{K}a_{k}(t)\leq 1 ,\, \, \, a_{k}(t)\geq \epsilon\,,k \in \mathbb{K},t \in \mathbb{T}.\tag{\theequation a}\label{equ:data1_a}\\
&0\leq f_{k}(t)\leq f_{max}, 0\leq p_{tm,k}(t)\leq P_{max},\nonumber\\&k \in \mathbb{K},t \in \mathbb{T}.\tag{\theequation b}\label{equ:data1_b}\\
&\underset{T\rightarrow \infty }{lim}\frac{{\rm E}[\left | Q_{k}(t) \right |]}{T}=0\,,k \in \mathbb{K}.\tag{\theequation c}\label{equ:data1_c}
\end{flalign}
Eq. (\ref{equ:data1_a}) is the bandwidth allocation constraint, where $\epsilon$ is a system constant. Constraints (\ref{equ:data1_b}) indicates the boundary of processor frequency and transmission power. For delay consideration, constraint (\ref{equ:data1_c}) forces the edge data buffers to be stable, which guarantees the collected data can be processed in a finite time. Among the constraints, index $k$ belongs to set $\mathbb{K}$ and time slot $t$ belongs to set $\mathbb{T}$

$\mathcal{P}_{2\text{-}\rm{A}}$ is obviously a statistical optimization problem with randomly arriving data. Therefore, the policy $\mathbf{\Phi }(t)$ has to be determined dynamically in each time slot. Furthermore, the spatial coupling of bandwidth allocation among UAV-BSs induces great challenge to the problem solution. Instead of solving $\mathcal{P}_{2\text{-}\rm{A}}$ directly, we propose an online jointly resource management algorithm based on Lyapunov optimization.

\subsection{Online optimization framework}
The proposed $\mathcal{P}_{2\text{-}\rm{A}}$ is a challenging statistical optimization problem. By Lyapunov optimization \cite{7274642}, $\mathcal{P}_{2\text{-}\rm{A}}$ can be formulated into a deterministic problem for each time slot, which can be solved with low complexity. The online algorithm can cope with the dynamical random environment while deriving an overall optimal outcome. Based on Lyapunov optimization framework ,the algorithm aims at saving energy while stabilizing the edge data buffers.

The Lyapunov function for time slot $t$ is defined as
\begin{equation}
L(t)=\frac{1}{2}\sum_{k=1}^{K}Q_{k}^{2}(t)
\end{equation}
This quadratic function is a scalar measure of data accumulation in queue. Its corresponding Lyapunov drift is defined as follows.
\begin{equation}
\Delta L(t)={\rm E}[L(t+1)-L(t)]
\end{equation}
To stabilize the network queuing buffer while minimizing the average energy penalty, the policy is determined by minimizing a bound on the following drift-plus-penalty function for each time slot $t$.
\begin{equation}\label{dpp}
\Delta _{V}(t)=\Delta L(t)+V\sum_{k=1}^{K}w_{k}P_{k}(t)
\end{equation}
where $V$ is a positive system parameter which represents the tradeoff between Lyapunov drift and energy cost. $\Delta L(t)$ is the expectation of a random process with unknown probability distribution. Therefore, an upper bound of $\Delta L(t)$ is estimated so that we can minimize $\Delta_{V}(t)$ without the specific probability distribution. According to the following Lemma \ref{upbound}, we derive a deterministic upper bound of $\Delta L(t)$ for each time slot.
\begin{lemma}\label{upbound}
For an arbitrary policy $\Phi (t)$ constrained by (\ref{equ:data1_a}), (\ref{equ:data1_b}) and (\ref{equ:data1_c}), the Lyapunov drift function is upper bounded by
\begin{equation}\label{upb}
\Delta L(t) \leq -\sum_{k=1}^{K}Q_{k}(t)(D_{l,k}(t)+D_{tm,k}(t))+C_{lp}
\end{equation}
where $C_{lp}$ is a known constant independent with the system policy and $Q_{k}(t)$ is the current data buffer length. $D_{l,k}(t)$ is the edge processing data bits while $D_{tm,k}(t)$ is the offloaded data bits. They are all for time slot $t$.
\end{lemma}
\begin{proof}
From equation (\ref{qup}), we have

\begin{align}
Q_{k}^{2}(t+1) &= ({\rm max}\{ Q_{k}(t)+A_{u,k}(t)\nonumber \\&-(D_{l,k}(t)+D_{tm,k}(t)),0 \})^{2}\nonumber  \\
& \leq (Q_{k}(t)+A_{u,k}(t)-(D_{l,k}(t)+D_{tm,k}(t)))^{2} \nonumber \\
&=Q_{k}^{2}(t)-2Q_{k}(t)(D_{l,k}(t)+D_{tm,k}(t)-A_{u,k}(t)) \nonumber \\
&+(D_{l,k}(t)+D_{tm,k}(t)-A_{u,k}(t))^{2} \label{drift1}
\end{align}
By (\ref{drift1}), we can subtract $Q_{k}^{2}(t)$ on both side and sum up the inequalities for $k=1, 2, ......, K$, which leads to follows.
\begin{align}
&\frac{1}{2}\sum_{k=1}^{K}\left [ Q_{k}^{2}(t+1)-Q_{k}^{2}(t) \right ] \nonumber \\
&\leq -\sum_{k=1}^{K}Q_{k}(t)(D_{l,k}(t)+D_{tm,k}(t))+\nonumber\\
&\sum_{k=1}^{K}\frac{(D_{l,k}(t)+D_{tm,k}(t)-A_{u,k}(t))^{2}}{2}+\sum_{k=1}^{K}Q_{k}(t)A_{u,k}(t)
\end{align}
As stated in Section \uppercase\expandafter{\romannumeral2}, the data rate of sensors is bounded by $[0,A_{max}]$. Furthermore, the channel capacity between sensors and UAV-BSs is also limited. Therefore, $A_{u,k}(t)$ is supposed to be upper bounded by $A_{u,max}$. Note that the computation and communication resources are limited. Therefore, $D_{l,k}(t)$ and $D_{tm,k}(t)$ are also bounded by their corresponding maximum processing rate. As the maximum processor frequency is $f_{max}$, we have $0 \leq D_{l,k}(t) \leq \frac{\tau f_{max}}{L_{k}}$. Since ${\rm log}_{2}(1+x) \leq \frac{x}{{\rm ln}2}$ and $p_{tm,k}(t) \in [0,P_{max}]$, we have $0 \leq D_{tm,k}(t) \leq \frac{\tau}{N_{0}}P_{max}\gamma_{k}g_{0}(\frac{d_{0}}{d_{k}})^{\theta}$. For simplicity, we separately denote $\frac{\tau f_{max}}{L_{k}}$ and $\frac{\tau}{N_{0}}P_{max}\gamma_{k}g_{0}(\frac{d_{0}}{d_{k}})^{\theta}$ as $D_{l,k,max}$ and $D_{tm,k,max}$. Then the term $(D_{l,k}(t)+D_{tm,k}(t)-A_{u,k}(t))^{2}$ should be bounded by ${\rm max}\{ A_{u,max}^{2}, (D_{l,k,max}+D_{tm,k,max})^{2} \}$
Therefore, we have
\begin{align}
&\frac{1}{2}\sum_{k=1}^{K}\left [ Q_{k}^{2}(t+1)-Q_{k}^{2}(t) \right ] \nonumber \\
&\leq -\sum_{k=1}^{K}Q_{k}(t)(D_{l,k}(t)+D_{tm,k}(t))+\nonumber\\
&\sum_{k=1}^{K}\frac{{\rm max}\{ A_{u,max}^{2}, (D_{l,k,max}+D_{tm,k,max})^{2} \}}{2}\nonumber\\&+\sum_{k=1}^{K}Q_{k}(t)A_{u,k}(t)\nonumber\\
&= -\sum_{k=1}^{K}Q_{k}(t)(D_{l,k}(t)+D_{tm,k}(t))+C_{lp}
\end{align}
where $C_{lp}=\sum_{k=1}^{K}\frac{{\rm max}\{ A_{u,max}^{2}, (D_{l,k,max}+D_{tm,k,max})^{2} \}}{2}+\sum_{k=1}^{K}Q_{k}(t)A_{u,k}(t)$. When considering a specific time slot $t$, it is straightforward to see that $C_{lp}$ is a deterministic constant. This completes the proof.
\end{proof}
Together with (\ref{dpp}) and (\ref{upb}), the drift-plus penalty function is upper-bounded by
\begin{align}
\Delta _{V}(t)\leq -\sum_{k=1}^{K}Q_{k}(t)(D_{l,k}(t)+D_{tm,k}(t))+V\sum_{k=1}^{K}w_{k}P_{k}(t)\nonumber\\+C_{lp}
\end{align}
By optimizing the above upper bound of $\Delta_{V}(t)$ in each time slot $t$, the data queuing length can be stabilized on a low level while the power consumption penalty is also minimized. In this way, the overall optimal policy can be derived without specific probability distributions.
In Lemma \ref{upbound}, parameter $C_{lp}$ is not affected by system policy. Therefore, it is reasonable to omit $C_{lp}$ in the policy determination problem.

Then the modified problem in each time slot $t$ based on Lyapunov optimization framework is defined as follows.
\begin{flalign}\label{equ:data2}
\mathcal{P}_{2\text{-}\rm{B}}:\,\,\min_{\mathbf{\Phi }(t)} \,\,\, & -\sum_{k=1}^{K}Q_{k}(t)(D_{l,k}(t)+D_{tm,k}(t))\nonumber\\&+V\sum_{k=1}^{K}w_{k}P_{k}(t)\\
\textrm{s.t.}\,\,\,& \sum_{k=1}^{K}a_{k}(t)\leq 1 ,\, \, \, a_{k}(t)\geq \epsilon\,,k \in \mathbb{K}\,\,, t \in \mathbb{T}.\tag{\theequation a}\label{equ:data2_a}\\
&0\leq f_{k}(t)\leq f_{max}, 0\leq p_{tm,k}(t)\leq P_{max},\nonumber\\&k \in \mathbb{K}\,\,, t \in \mathbb{T}.\tag{\theequation b}\label{equ:data2_b}
\end{flalign}

%

\subsection{Solution for $\mathcal{P}_{2\text{-}\rm{B}}$}
In last subsection, we formulated $\mathcal{P}_{2\text{-}\rm{B}}$ for deriving optimal policy in each time slot. The optimization objectives include local computation processor frequency $\mathbf{f}(t)$, data transmission power $\mathbf{p}_{tm}(t)$ and bandwidth allocation $\mathbf{a}(t)$. In this section, we will divide $\mathcal{P}_{2\text{-}\rm{B}}$ into two subproblems and derive a solution for optimal policy.
\subsubsection{Optimal frequency for edge processor}
We first delete part of the objective function independent of $\mathbf{f}(t)$. Then it is straightforward to see that the subproblem with respect to $\mathbf{f}(t)$ is defined as follows.
\begin{flalign}\label{subp1}
\mathcal{P}_{3\text{-}\rm{A}}:\,\,\min_{\mathbf{f }(t)} \,\,\, & -\sum_{k=1}^{K}\frac{\tau Q_{k}(t)}{L_{k}}f_{k}(t)+V\sum_{k=1}^{K}w_{k}\kappa _{k}f_{k}^{3}(t)\\
\textrm{s.t.}\,\,\,&0\leq f_{k}(t)\leq f_{max}\,,k \in \mathbb{K}\,\,, t \in \mathbb{T}.\tag{\theequation a}\label{equ:subp1_a}
\end{flalign}

It is obvious to confirm that $\mathcal{P}_{3\text{-}\rm{A}}$ is a convex optimization problem. Furthermore, there is no coupling among elements in $\mathbf{f }(t)$. Therefore, the optimal processor frequency can be derived separately for each $k$. The stationary point of $\frac{\tau Q_{k}(t)}{L_{k}}f_{k}(t)+Vw_{k}\kappa _{k}f_{k}^{3}(t)$ is $\sqrt{\frac{\tau Q_{k}(t)}{3L_{k}w_{k}\kappa_{k} V}}$. In addition, the optimal processor frequency may also be the boundary $f_{max}$. Then the final solution is given by
\begin{equation}\label{fopt}
f_{k}^{*}(t)={\rm min} \{ f_{max}, \sqrt{\frac{\tau Q_{k}(t)}{3L_{k}w_{k}\kappa_{k} V}} \}\,\,\,(w_{k}>0,V>0)
\end{equation}

\begin{remark}
The optimal processor frequency $f_{k}^{*}(t)$ is a monotone increasing function with respect to data queuing length $Q_{k}(t)$. A straightforward insight is that edge servers tend to process faster as there is much data accumulating in the data buffer. Besides, as $V$ or $w_{k}$ increases, the proportion of edge computation energy cost becomes larger, which results in decreasing of processor frequency. As parameter $\kappa_{k}$ increases, the energy consumption per-frequency gets larger, which causes $f_{k}^{*}(t)$ to decrease. Furthermore, a larger $L_{k}$ corresponds to a lower edge processing frequency. Then the edge server should lower down its processor frequency and offload more data to the cloud.
\end{remark}

\subsubsection{Bandwidth allocation and data transmission power}
We reserve the elements with respect to $\mathbf{p}_{tm}(t)$ and $\mathbf{a}(t)$ and derive the following subproblem.
\begin{flalign}\label{subp2}
\mathcal{P}_{3\text{-}\rm{B}}:\,\,\min_{\mathbf{p_{tm} }(t),\mathbf{a }(t)} \,\,\, & -\sum_{k=1}^{K}Q_{k}(t)D_{tm,k}(t)+V\sum_{k=1}^{K}w_{k}p_{tm,k}(t)\\
\textrm{s.t.}\,\,\,& \sum_{k=1}^{K}a_{k}(t)\leq 1 ,\, \, \, a_{k}(t)\geq \epsilon\,,k \in \mathbb{K}\,\,, t \in \mathbb{T}.\tag{\theequation a}\label{equ:data2_a}\\
&0\leq p_{tm,k}(t)\leq P_{max},k \in \mathbb{K}\,\,, t \in \mathbb{T}.\tag{\theequation b}\label{equ:data2_b}
\end{flalign}

In (\ref{subp2}), we have $D_{tm,k}(t)=a_{k}(t)W\tau{\rm log}_{2}(1+\frac{\Gamma_{k}(t)p_{tm,k}(t)}{a_{k}(t)N_{0}W})$. Note that this is a perspective function of $\widetilde{D}(p_{tm}(t))=W\tau{\rm log}_{2}(1+\frac{\Gamma_{k}(t)p_{tm,k}(t)}{N_{0}W})$ with $D_{tm,k}(t)=a_{k}(t)\widetilde{D}(p_{tm}(t)/a_{k}(t))$. It is straightforward to see that $\widetilde{D}(p_{tm}(t))$ is a concave function with respect to $p_{tm}(t)$. Then $D_{tm,k}(t)$ is jointly concave with respect to $a_{k}(t)$ and $p_{tm,k}(t)$. Therefore, $\mathcal{P}_{3\text{-}\rm{B}}$ is a convex optimization problem. Though it can be solved directly by conventional solvers, the dimensional curse may still be a large obstacle. In this paper, we employ dual decomposition and sequential optimization to solve $\mathcal{P}_{3\text{-}\rm{B}}$ in a more efficient way. Note that the optimal solution of $\mathbf{p_{tm} }(t)$ and $\mathbf{a }(t)$ are coupled to each other. They will be separately solved supposing the other one is fixed. By iteratively optimizing $\mathbf{p_{tm} }(t)$ and $\mathbf{a }(t)$ in turns, the optimal policy in $\mathcal{P}_{3\text{-}\rm{B}}$ will be derived.

Suppose the bandwidth allocation $\mathbf{a }(t)$ is fixed, $\mathcal{P}_{3\text{-}\rm{B}}$ can be reformulated as follows.
\begin{flalign}\label{subp3}
\mathcal{P}_{3\text{-}\rm{C}}:\,\,\min_{\mathbf{p_{tm} }(t)} \,\,\, & -\sum_{k=1}^{K}Q_{k}(t)a_{k}(t)W\tau {\rm log}_{2}(1+\frac{\Gamma_{k}(t)p_{tm,k}(t)}{a_{k}(t)N_{0}W})\nonumber\\&+V\sum_{k=1}^{K}w_{k}p_{tm,k}(t)\\
\textrm{s.t.}\,\,\,&0\leq p_{tm,k}(t)\leq P_{max},k \in \mathbb{K}\,\,, t \in \mathbb{T}.\tag{\theequation a}\label{equ:data3_a}
\end{flalign}
As bandwidth allocation $\mathbf{a }(t)$ is fixed, it is straightforward to see that elements in $\mathbf{p_{tm} }(t)$ are not coupled with each other. Therefore, the optimal transmission power $\mathbf{p}_{tm,k}^{*}(t)$ can be obtained independently for each index $k$. Decomposing $\mathcal{P}_{3\text{-}\rm{C}}$ by index $k$, the corresponding
optimization objective is
\begin{equation}\label{obj}
Q_{k}(t)a_{k}(t)W\tau {\rm log}_{2}(1+\frac{\Gamma_{k}(t)p_{tm,k}(t)}{a_{k}(t)N_{0}W})+Vw_{k}p_{tm,k}(t)
\end{equation}
The stationary point of (\ref{obj}) is $a_{k}(t)W[\frac{Q_{k}(t)\tau}{Vw_{k}{\rm ln}2}-\frac{N_{0}}{\Gamma_{k}(t)}]$. As $p_{tm,k}(t)$ is bounded by $[0,P_{max}]$, the final optimal transmission power is
\begin{equation}
p_{tm,k}^{*}(t)={\rm min}\{ {\rm max}\{ a_{k}(t)W[\frac{Q_{k}(t)\tau}{Vw_{k}{\rm ln}2}-\frac{N_{0}}{\Gamma_{k}(t)}],0 \},P_{max} \}
\end{equation}

Since optimized $\mathbf{p_{tm} }(t)$ with fixed $\mathbf{a }(t)$ has been obtained, it is straightforward to continue optimizing $\mathbf{a }(t)$ with fixed $\mathbf{p_{tm} }(t)$. The corresponding sub-problem is defined as follows.

\begin{flalign}\label{subp4}
\mathcal{P}_{3\text{-}\rm{D}}:\,\,\min_{\mathbf{a}(t)} \,\,\, & -\sum_{k=1}^{K}Q_{k}(t)a_{k}(t)W\tau {\rm log}_{2}(1+\frac{\Gamma_{k}(t)p_{tm,k}(t)}{N_{0}Wa_{k}(t)})\\
\textrm{s.t.}\,\,\,&\sum_{k=1}^{K}a_{k}(t)\leq 1 ,\, \, \, a_{k}(t)\geq \epsilon\,,k \in \mathbb{K}\,\,, t \in \mathbb{T}.\tag{\theequation a}\label{equ:subp4_a}
\end{flalign}

The optimal solutions $a_{k}^{*}(t)$ are coupled in terms of index $k$ by constraints (\ref{equ:subp4_a}). Therefore, $\mathcal{P}_{3\text{-}\rm{D}}$ can not be directly decomposed into sub-problems. In this case, dual decomposition method is applied to obtain decoupled sub-problems of $\mathcal{P}_{3\text{-}\rm{D}}$.

The Lagrange function of $\mathcal{P}_{3\text{-}\rm{D}}$ is
\begin{flalign}\label{lagrange1}
L(\mathbf{a}(t),\lambda)&=-\sum_{k=1}^{K}Q_{k}(t)a_{k}(t)W\tau {\rm log}_{2}(1+\frac{\Gamma_{k}(t)p_{tm,k}(t)}{N_{0}Wa_{k}(t)})\nonumber\\&+\lambda(\sum_{k=1}^{K}a_{k}(t)-1)
\end{flalign}
In (\ref{lagrange1}), $a_{k}(t)$ has been decoupled. Then by minimizing $L(\mathbf{a}(t),\lambda)$ with respect to $a_{k}(t)$, the dual function for index $k$ is derived as
\begin{flalign}\label{dual-d}
g_{k}(\lambda)=&\underset{a_{k}(t)\geq \epsilon }{inf}(-Q_{k}(t)a_{k}(t)W\tau {\rm log}_{2}(1+\frac{\Gamma_{k}(t)p_{tm,k}(t)}{N_{0}Wa_{k}(t)})\nonumber\\&+\lambda a_{k}(t))
\end{flalign}
From (\ref{dual-d}), the dual function of $\mathcal{P}_{3\text{-}\rm{D}}$ is
\begin{equation}\label{dual-sum}
L(\lambda)=-\lambda +\sum_{k=1}^{K}g_{k}(\lambda)
\end{equation}
Finally, the dual problem of $\mathcal{P}_{3\text{-}\rm{D}}$ is

\begin{flalign}\label{subp5}
\mathcal{P}_{3\text{-}\rm{E}}:\,\,\max_{\lambda}\,\,\, & -\lambda +\sum_{k=1}^{K}g_{k}(\lambda)\,\,\, & \\
\textrm{s.t.}\,\,\,&\lambda \geq 0.\tag{\theequation a}\label{equ:subp5_a}
\end{flalign}
Dual sub-problem $\mathcal{P}_{3\text{-}\rm{E}}$ can be solved by gradient decent method. According to (\ref{dual-sum}) and (\ref{dual-d}), the corresponding gradient is
\begin{equation}\label{gradient}
\bigtriangledown L(\lambda )=\sum_{k=1}^{K}a_{k}^{*}(t)-1
\end{equation}
where $a_{k}^{*}(t)$ is the optimal bandwidth allocation of index $k$ for current $\lambda$, which achieves the lower bound in (\ref{dual-d}).

\begin{figure}[tbp]
  \centering
  \includegraphics[width=0.92\columnwidth]{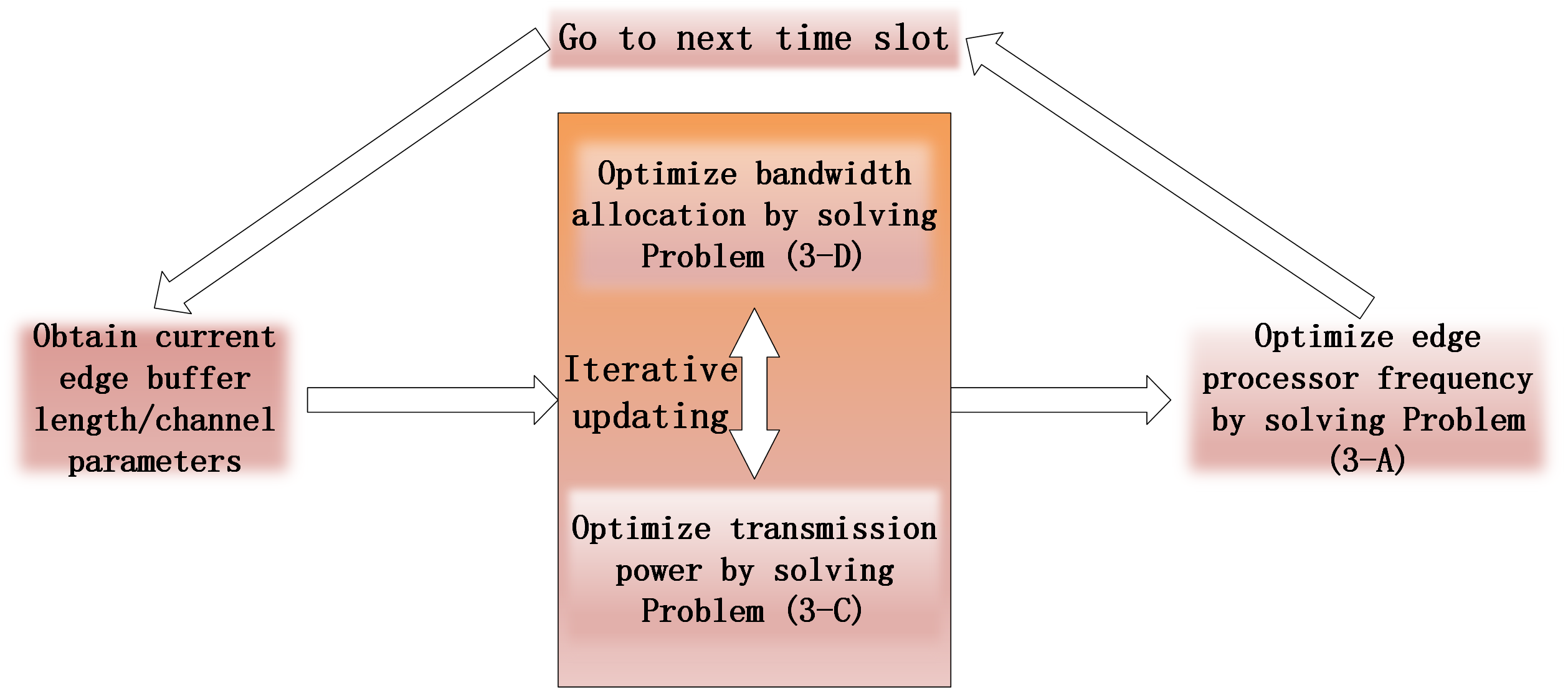}
  \caption{The overall process to solve $\mathcal{P}_{2\text{-}\rm{B}}$, which involves the solution of all sub-problems.}\label{online_process}
\end{figure}

Note that this is a convex optimization problem and $a_{k}(t)$ is constrained by $a_{k}(t) \geq \epsilon$. Therefore, either the stationary point or $\epsilon$ achieves the lower bound in (\ref{dual-d}). It is straightforward to see that the stationary point can be derived by conducting bisection method on the derivative function. Denoting the derived stationary point as $a_{k}^{s}(t)$, we have
\begin{equation}
a_{k}^{*}(t)= {\rm max}\{ a_{k}^{s}(t), \epsilon \}
\end{equation}
Together with (\ref{gradient}), $\lambda$ is updated by
\begin{equation}\label{lupd}
\lambda ^{(n+1)}=\lambda ^{(n)}+\sum_{k=1}^{K}a_{k}^{*}(t)-1
\end{equation}
where $n$ is the number of iterations.
By iteratively updating $\lambda$ and computing corresponding $a_{k}^{*}(t)$, we can finally derive the optimal bandwidth allocation and transmission power.
\begin{remark}
The insights of the iteration method can be explained as follows.
Given $\lambda^{(n)}$, if its corresponding bandwidth allocation happens to satisfy $\sum_{k=1}^{K}a_{k}^{*}(t) > 1$, $\lambda^{(n+1)}$ will increase according to (\ref{lupd}). This will obviously leads $g_{k}(\lambda)$ to increase. Furthermore, from definitions in (\ref{lagrange1}), $\frac{\partial L(\mathbf{a}(t),\lambda)}{\partial \mathbf{a}(t)}$ will tend to be positive when $\lambda$ is sufficiently large. In this case, the corresponding stationary point $a_{k}^{s}(t)$ will be smaller. Then the derived $\sum_{k=1}^{K}a_{k}^{*}(t)$ with respect to a larger $\lambda^{(n+1)}$ will tend to decrease. Meanwhile, if $\sum_{k=1}^{K}a_{k}^{*}(t) < 1$, $\lambda^{(n+1)}$ will decrease, which leads $\sum_{k=1}^{K}a_{k}^{*}(t)$ to increase. Finally, the iterations will lead to a desired bandwidth allocation.
\end{remark}

In summary, the final solution to $\mathcal{P}_{2\text{-}\rm{B}}$ is shown as the chart in Fig. \ref{online_process}. All the discussed sub-problems are combined to obtain an optimal policy $\mathbf{\Phi }(t)$.

\section{Simulations}
We carried out simulations of the distributed data processing network to test the performance of proposed algorithms for network management. It is assumed that $20000$ sensors are distributed in a $600\times 400$ area. The distributed data generation satisfies poisson distribution with rate $\lambda_{l}$ for $d_{l}$. Rate $\lambda_{l}$ is supposed to be uniformly distributed in $[250,300]$. Its communication rate with UAV-BSs is set as $2000$ bits/s. The system time slot interval $\tau$ is set as $0.5$s. The $K$ UAV-BSs start hovering from randomly distributed positions in the $600\times 400$ area. The radius of UAV-BS coverage is $60$. Their action set for path update is $\Psi$ in (\ref{equ:policy}), where $v$ is set as 8. The path update interval $\tau_{p}$ is $5\tau$ and $\Phi_{u}$ in observation processing is set as $8000$. The parameters with respect to edge data processing are $f_{max}=2{\rm GHz}$, $\kappa_{k}=10^{-26}$, $L_{k}=3000 {\rm Cycles/bit}$. The data offloading channel involves $g_{0}=10^{-4}$, $d_{0}=1$, $\theta=4$, $W=2{\rm MHz}$, $P_{max}=5{\rm W}$, $N_{0}=-167{\rm dBm/}$ and $\gamma_{k}(t)\sim {\rm Exp(1)}$. The weights $w_{k}$ in (\ref{sump}) is set as $\frac{1}{K}$ \cite{mao2017stochastic}.

In simulations of the path planning algorithm, we apply a CNN network with four hidden layers \cite{mnih2015human}. The input data is an $84\times 84\times 1$ observation produced by Algorithm \ref{alg13}. The first hidden layer consists of 32 filters of $8\times 8$ with stride $4$. The second hidden layer consists of 64 filters of $4\times 4$ with stride 2. The third hidden layer consists of 64 filters of $3\times 3$ with stride 1. The final hidden layer is fully-connected with 512 units. Each of the hidden layer is followed by a nonlinear rectifier function \cite{nair2010rectified}. The output layer is fully-connected with an estimated reward value for each candidate action. In training process, we apply $\epsilon$-greedy strategy as shown in Algorithm \ref{alg11}. The original $\epsilon(0)$ is set as $0.97$. In each training, the coefficient will decay by $\epsilon(t_{p}+1)=0.92\epsilon(t_{p})$. $\epsilon(t_{p})$ is reset as $0.97$ every 300 time slots so that the system can keep exploring around environment and learn the new explorations.

As shown in Fig. \ref{stable_test} ,we first validate the effective coverage of the proposed path planning algorithm based on deep reinforcement learning. Starting from random initial positions, the average service urgency $\frac{1}{L}\sum_{l} \phi_{l}$ is recorded within $10000$ system time slots. Parameter $\phi_{l}$ for sensor $d_{l}$ is defined in Algorithm \ref{alg12}, which reflects the waiting time of $d_{l}$ before covered. Therefore, a small and stable $\frac{1}{L}\sum_{l} \phi_{l}$ corresponds to a better coverage performance. In Fig. \ref{stable1}, the performance of randomly selected policy with $K=6$ is marked by '$>$'. Its average service urgency is the largest with the worst stability. The proposed path planning algorithm with $K=6$ and $168\times 168$ observation range corresponds to the curve marked by circles. Compared with the random policy, it obtains a much smaller service urgency with enhanced stability. Setting $K=9$, the coverage performance is further improved as shown by curves marked by squares. Finally, the observation range is set as $252\times 252$ with $K=9$. As shown by curves marked by '*', the coverage is brilliant with the best stability. In Fig. \ref{stable2}, we validate the influence of balance on $\alpha$ in (\ref{upalpha}). Without balance on $\alpha$, the training update rate is set to be equal for each candidate action. The result is shown by curves marked by circles. It is straightforward to see that its stability is much worse compared with the proposed adaptive balance on $\alpha$, especially for the starting phase. In learning process, the frequently occurring actions will be trained more, which results in challenge to training convergence. In multi-agent system, this issue will be enlarged. Therefore, the balance on $\alpha$ will efficiently enhance performance of reinforcement learning.

\begin{figure}[tbp]
\centering
\subfigure[Test of average $\phi_{l}$ for $4$ cases in $10000$ time slots. ] { \label{stable1}
\includegraphics[width=0.8\columnwidth]{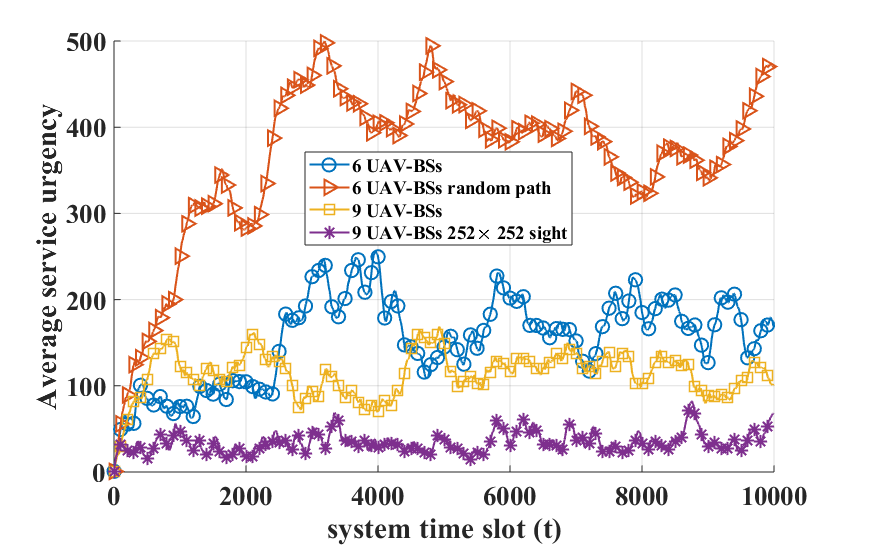}}\\
\subfigure[Test on the balance of $\alpha$ for candidate actions ($K=9\,,252\times 252$).]{ \label{stable2}
\includegraphics[width=0.8\columnwidth]{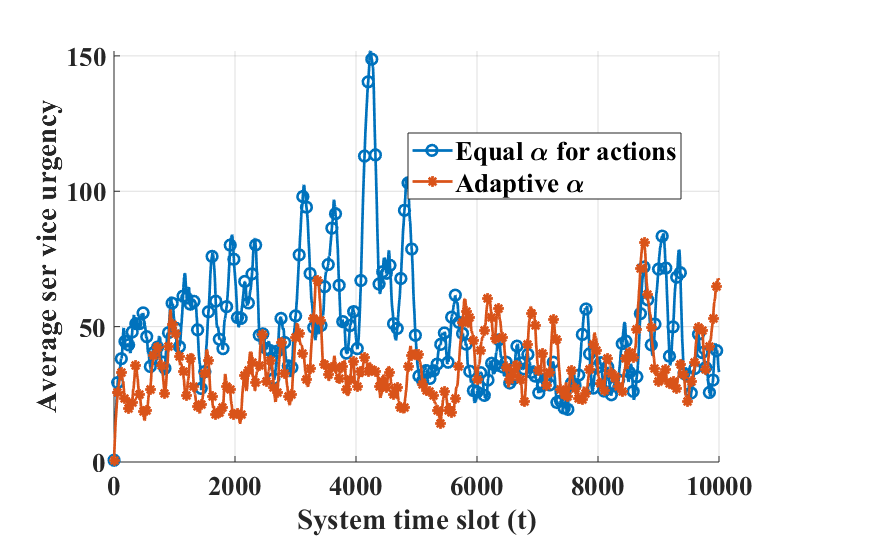}}
\caption{The performance in terms of average service urgency for random path planning, 6-UAV case, 9-UAV case and 9-UAV case with larger observation range.} \label{stable_test}
\end{figure}

In Fig. \ref{gamma_test}, we investigate the coverage performance with respect to parameter $\gamma$. It represents the decay of future rewards in overall rewards. A smaller $\gamma$ means that we only care rewards in a short time range. A larger $\gamma$ means the future rewards in a longer time range is taken into consideration. In simulations, we record the mean value and variance of $\frac{1}{L}\sum_{l} \phi_{l}$ within $3000$ time slots for varied $\gamma$. Fig. \ref{mean_gamma} displays the mean value of $\frac{1}{L}\sum_{l} \phi_{l}$ for $252\times 252$ sight and $168\times 168$ sight, while Fig. \ref{var_gamma} shows the corresponding variance. The results are derived by Monte Carlo Method with $12$ experiments for each $\gamma$. As shown in Fig. \ref{mean_gamma} and Fig. \ref{var_gamma}, the increase of $\gamma$ will derive a better coverage. In this case, the planner will consider more future elements and enhance its policy. However, both curves meet the turning point as $\gamma$ gets larger. In this case, the planner considers a long range of future rewards, which is beyond the local observations. Therefore, the coverage may get worse. Note that for relatively small $\gamma$, curves of $252\times 252$ observation range does not show improvement of coverage. In this situation, the planner considers little future elements, which leads to the poor utilization of information in a larger observation range. Simulation results show that $0.8$ is a reasonable value for $\gamma$.

\begin{figure}[tbp]
\centering
\subfigure[The overall sensor service urgency with respect to future rewards decay $\gamma$. ] { \label{mean_gamma}
\includegraphics[width=0.7\columnwidth]{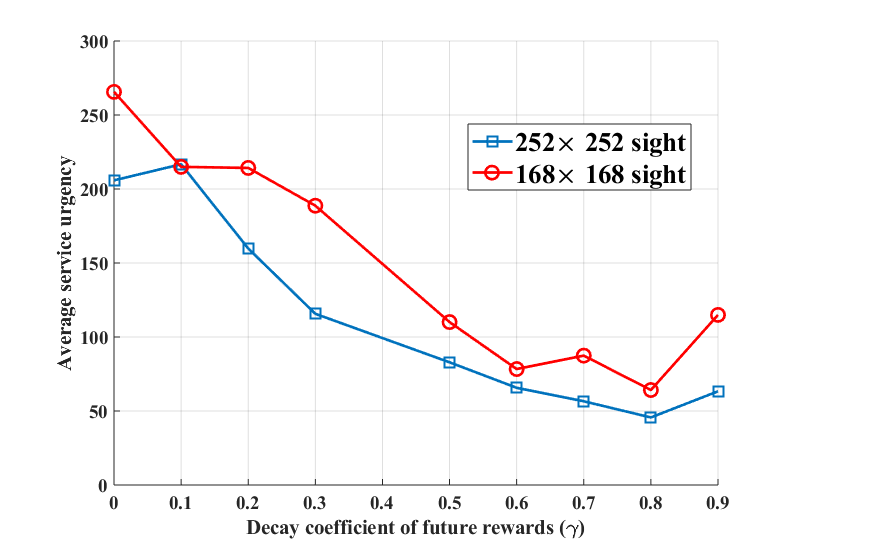}}\\
\subfigure[The entire variance of service urgency with respect to coefficient $\gamma$.]{ \label{var_gamma}
\includegraphics[width=0.7\columnwidth]{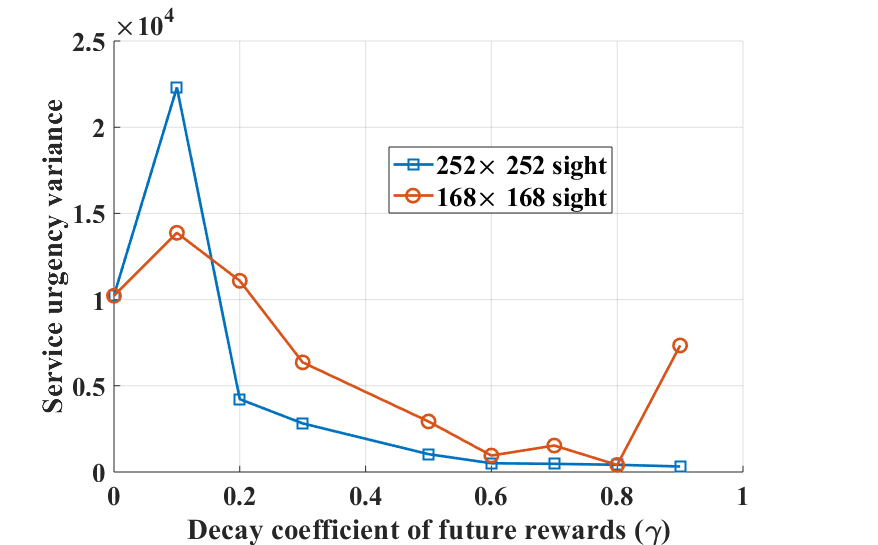}}
\caption{The coverage stableness and mean value with respect to coefficient $\gamma$ for $168 \times 168$ and $252\times 252$ observation range.} \label{gamma_test}
\end{figure}

In Fig. \ref{edge_length}, we record the average UAV buffer length within $10000$ time slots, which represents the edge data delay. Parameter $K$ is set as $9$ with $252 \times 252$ observations. As shown in Fig. \ref{edge_length2}, the system turns out to break down if we only apply edge data processing or only transmit collected data to the cloud. Faced with the huge data set, the edge processor and the communication network will be too stressful to maintain the system. Then the data will keep accumulating until the system breaks down. In Fig. \ref{edge_length1}, the curve marked by squares represents the proposed data management algorithm based on Lyapunov optimization. It is compared with the evenly bandwidth allocation method, where $a_{k}(t)$ is evenly set as $\frac{1}{K}$. As shown in Fig. \ref{edge_length1}, the proposed data management algorithm performs much better than evenly allocating bandwidth in dealing with the big data circumstances. In this situations, the proposed algorithm can smartly allocate the network resources to balance the varied burden on edge nodes.

\begin{figure}[tbp]
\centering
\subfigure[The separate performance in terms of edge buffer length. ] { \label{edge_length2}
\includegraphics[width=0.75\columnwidth]{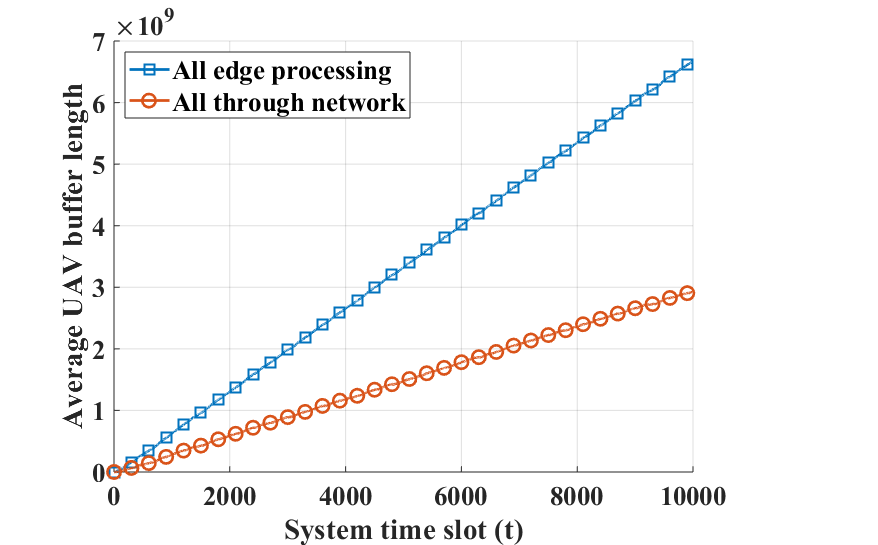}}\\
\subfigure[The performance of single edge processing and single data transmission system.]{ \label{edge_length1}
\includegraphics[width=0.75\columnwidth]{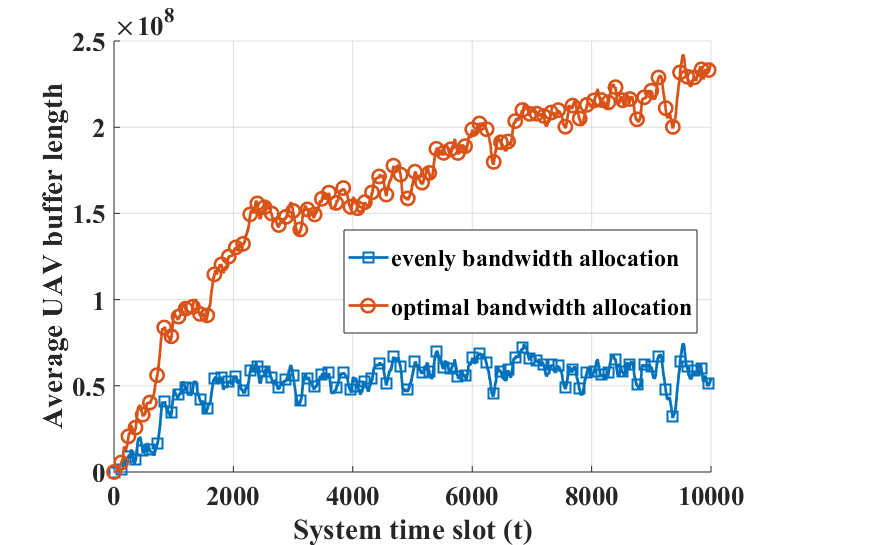}}
\caption{The record of average edge buffer length for optimal network management policy, evenly bandwidth allocation, single edge processing and single data transmission mode ($K=9$, $252\times 252$ sight, $V=6\times 10^{9}$).} \label{edge_length}
\end{figure}

\begin{figure}[tbp]
\centering
\subfigure[The performance in terms of average power consumption. ] { \label{power_V}
\includegraphics[width=0.75\columnwidth]{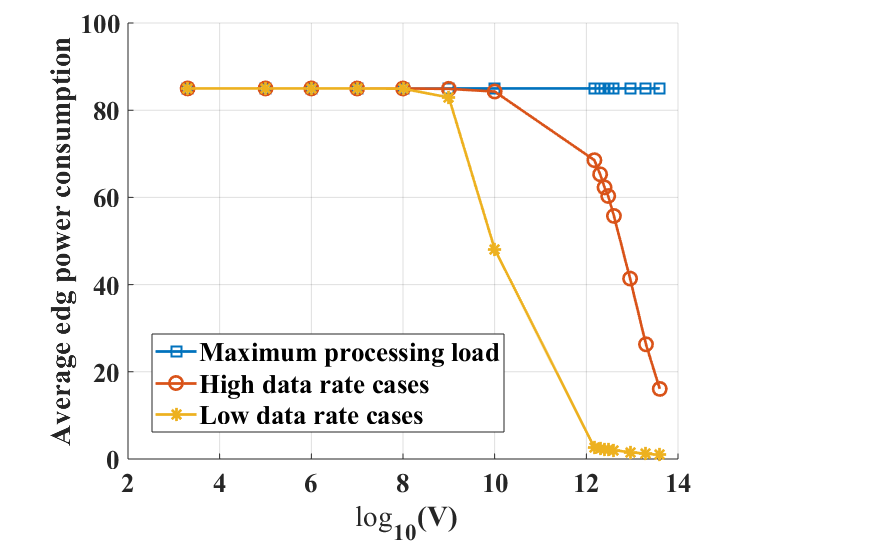}}\\
\subfigure[The recorded average buffer length.]{ \label{buffer_V}
\includegraphics[width=0.75\columnwidth]{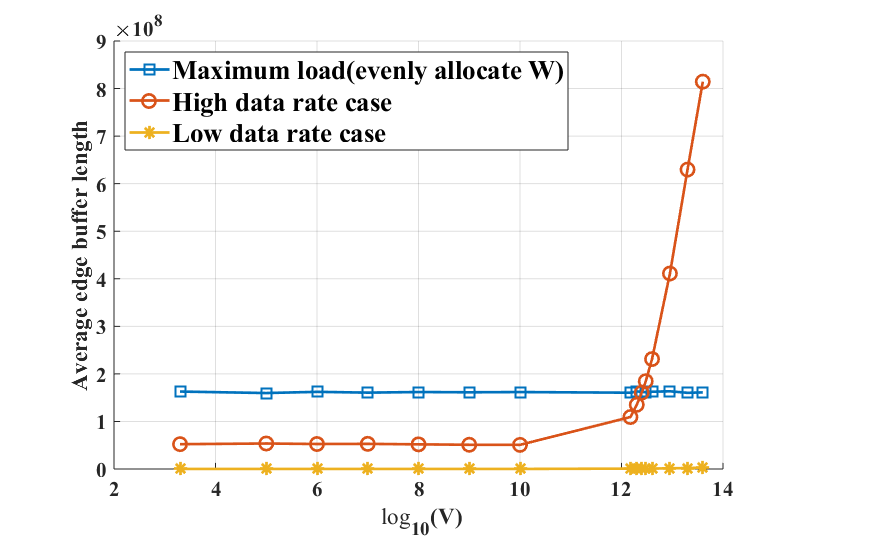}}
\caption{The average power consumption and buffer length for maximum processing load($f_{k}(t)=f_{max}, p_{tm,k}(t)=P_{max}$), proposed network management policy for equal data rate and low data rate ($K=9$, $252\times 252$ sight).} \label{V_test}
\end{figure}

Fig. \ref{V_test} shows the influence of parameter $V$ on power consumption and average UAV buffer length. The low data rate case means half the former data rate of sensors, while the high data rate case remains unchanged. They both apply the proposed data scheduling algorithm based on Lyapunov optimization. The maximum processing load case means setting $f_{k}(t)=f_{max}$, $p_{tm,k}(t)=P_{max}$ and $a_{k}(t)=\frac{1}{K}$. Its data rate is the same with the high data rate case. As shown in Fig. \ref{power_V}, the power consumption of the maximum processing load case remains at the top level. The proposed algorithm lowers down the power consumption as $V$ increases. In Fig. \ref{buffer_V}, the data processing delay is investigated. The maximum working load case is not affected by $V$. The proposed algorithm maintains good performance in low data rate case. In high data rate case, its data delay is kept at a low level unless $V$ gets too large. Taking large $V$, the algorithm will pose too much weight on power consumption, which results in large delay for high data rate case. Note that the proposed algorithm even achieves a lower delay compared with the maximum processing load case. That is caused by smartly determining $a_{k}(t)$ for bandwidth allocation. That is, the adaptive data scheduling algorithm can save energy while lower down delay.

\section{Conclusion}
In this paper, we investigated a big data processing system for applications in internet of things. The system is composed of three layers, involving distributed sensors, UAV-BSs and center cloud. To collect data distributed among sensors with efficiency, a UAV path planning algorithm based on deep reinforcement learning was proposed. The local observations by UAV-BSs were taken as input of neutral networks to predict rewards of candidate actions. The corresponding designing issues involving observation feature, training process and action rewards were figured out. By simulations, we validated the efficient coverage of the proposed path planning algorithm. To process the collected data with efficiency while saving power, we developed a network scheduling algorithm based on Lyapunov optimization. It figured out the network resources scheduling and achieved a tradeoff between edge pre-processing and network transmission. Its performance in terms of data delay and power consumption was tested by simulations. For the future, it would be interesting to extend the work to scenarios in smart cities, where the user behavior, mobility and its coexist with existing cellar network should be investigated.


%

%



\ifCLASSOPTIONcaptionsoff
  \newpage
\fi



\bibliography{usart}
\bibliographystyle{unsrt}

\end{document}